\documentclass[journal]{IEEEtran}
\usepackage{amsmath,amssymb,bbm}
\usepackage{amsthm}
\usepackage{multicol}
\usepackage{graphicx,epsfig}
\usepackage{subfigure}
\usepackage{epstopdf}
\usepackage{verbatim}
\usepackage{soul}
\usepackage{color}
\usepackage{mathptmx}
\usepackage{mathrsfs}
\usepackage{booktabs}
\usepackage{fancyhdr}
\usepackage{tabularx}
\usepackage{xfrac}
\usepackage{CJK}
\usepackage{times}
\usepackage{hhline}
\usepackage{multirow}
\usepackage{caption}
\usepackage{algorithm,algorithmic}
\usepackage{cite}
\usepackage[]{footmisc}
\usepackage[numbers,sort&compress]{natbib}
\newtheorem{definition}{Definition}
\newtheorem{theorem}{Theorem}
\newtheorem{lemma}{Lemma}
\newtheorem{remark}{Remark}

\newtheorem*{proof*}{Proof}

\begin{document}

\bstctlcite{IEEEexample:BSTcontrol}

\title{Region of Attraction for Power Systems using Gaussian Process and Converse Lyapunov Function -- Part I: Theoretical Framework and Off-line Study}

\author{Chao Zhai and Hung D. Nguyen \thanks{Chao Zhai and Hung D. Nguyen are with School of Electrical and Electronic Engineering, Nanyang Technological University, 50 Nanyang Avenue, Singapore 639798. Email:  zhaichao@ntu.edu.sg and hunghtd@ntu.edu.sg. Corresponding author: Hung D. Nguyen.}}


\maketitle

\begin{abstract}
This paper introduces a novel framework to construct the region of attraction (ROA) of a power system centered around a stable equilibrium by using stable state trajectories of system dynamics. Most existing works on estimating ROA rely on analytical Lyapunov functions, which are subject to two limitations: the analytic Lyapunov functions may not be always readily available, and the resulting ROA may be overly conservative. This work overcomes these two limitations by leveraging the converse Lyapunov theorem in control theory to eliminate the need for an analytic Lyapunov function and learning the unknown Lyapunov function with the Gaussian Process (GP) approach. In addition, a Gaussian Process Upper Confidence Bound (GP-UCB) based sampling algorithm is designed to reconcile the trade-off between the exploitation for enlarging the ROA and the exploration for reducing the uncertainty of sampling region. Within the constructed ROA, it is guaranteed in the probability that the system state will converge to the stable equilibrium with a confidence level. Numerical simulations are also conducted to validate the assessment approach for the ROA of the single machine infinite bus system and the New England $39$-bus system. Numerical results demonstrate that our approach can significantly enlarge the estimated ROA compared to that of the analytic Lyapunov counterpart.
\end{abstract}


\section*{Nomenclature}
\begin{IEEEdescription}[\IEEEusemathlabelsep\IEEEsetlabelwidth{$V_1$,$V_2$}]
\item[$\alpha(z)$] a class $\Gamma$ function
\item[$V(\mathbf{x})$] Lyapunov function
\item[$\hat{V}(\mathbf{x})$] estimation of Lyapunov function
\item[$V^\star(\mathbf{x})$] an existing Lyapunov function
\item[$\phi(\mathbf{x},t)$] state trajectory of nonlinear dynamical system
\item[$k(\mathbf{x},\mathbf{x}')$] covariance function of GP or kernel function in the reproducing kernel Hilbert space (RKHS)
\item[$S$] real region of attraction for nonlinear system
\item[$\Omega_c$] level set of Lyapunov function with the upper bound $c$
\item[$N$] number of stable sampling points.
\item[$\Delta t$] time interval of numerical method for solving the differential equation
\item[$\delta$] parameter to specify the confidence level
\item[$\mathcal{A}_N$] set of $N$ stable sampling points
\item[$\mu_i(\mathbf{x})$] mean value of a Gaussian process at the $i$-th iteration
\item[$\sigma_i(\mathbf{x})$] standard deviation of a Gaussian process at the $i$-th iteration
\item[$\|\cdot\|$] 2-norm in Euclidean space
\item[$\|\cdot\|_k$] induced RKHS norm with the kernel $k(\mathbf{x},\mathbf{x}')$
\end{IEEEdescription}

\section{Introduction}
The region of attraction (ROA) for complex dynamical systems provides a useful measure of stability level and robustness against external disturbances. Thus, it is of great importance to safety-critical systems (e.g., power systems, nuclear reactor control systems, engine control systems, etc), where the system stability has to be guaranteed before they are implemented in practice. In terms of power systems, the ROA refers to a subspace of operating states that can converge to a steady-state equilibrium. There are various approaches for estimating the ROA of a general nonlinear system such as contraction analysis \cite{nguyen2017contraction}, level sets of Lyapunov function \cite{kha96, Longfamily}, sum of square technique \cite{izu18,che11}, sampling-based method \cite{bob16}, and so on. Nevertheless, these approaches largely rely on deterministic models and may not be applicable to deal with uncertainties in more realistic systems.


As a non-parametric method, GP is flexible to incorporate the prior information as well as to quantify the uncertainty \cite{kan18}. By regarding the unknown function or dynamics as a GP,
Bayesian machine learning provides a powerful tool for both regression and inference using the prior belief and sample data, and it can generate the posterior distribution for the unknown function \cite{ras10, teti}. Therefore, the GP approach has found wide applications in various fields such as bandit setting \cite{sri12}, robotics \cite{ber16}, and classifier design \cite{gib00}, to name just a few. As is well known, it is always a challenging problem for determining a suitable Lyapunov function for a complex nonlinear dynamical system. For a stable equilibrium point of the dynamical system, the converse Lyapunov theorem ensures the existence of Lyapunov functions and enables us to compute the values using stable state trajectories \cite{kha96}. In practice, the accurate values of Lyapunov function are not available due to numerical error and time restrictions. For this reason, it would be desirable to regard the unknown Lyapunov function as a GP. In this way, the posterior distribution of Lyapunov function can be obtained by learning the sampling data (i.e., estimated values of Lyapunov function), which makes it possible to construct and evaluate the ROA for nonlinear dynamical systems.

Thus, this paper centers on the quantitative evaluation of ROA for power systems from a GP perspective. Compared with existing work \cite{mun13,ber16,jone17}, the key contributions of this work lie in
\begin{enumerate}
    \item Propose a theoretical framework based on the converse Lyapunov theorem for estimating the ROA of general nonlinear systems around an equilibrium without constructing an analytic Lyapunov function.
    \item Develop a GP-UCB based sampling algorithm for creating a sampling set and learning the unknown Lyapunov function using stable state trajectories.
    \item For an existing Lyapunov function, our approach allows for extending the certified ROA with a guaranteed confidence level.
\end{enumerate}

The remainder of this paper is organized as follows. Section \ref{sec:roa} introduces the region of attraction for a general dynamical system and the estimation of the Lyapunov function. Section \ref{sec:gp} presents the GP approach for learning the known dynamics, followed by the main results on the GP-UCB based sampling algorithm in Section \ref{sec:main}. Numerical simulations are conducted to validate the proposed approach on the IEEE test systems in Section \ref{sec:sim}. Finally, Section \ref{sec:con} draws a conclusion and discusses future work.

\section{The ROA of a General Dynamical System} \label{sec:roa}
The ROA of a general dynamical system normally refers to a region where each state can converge to the stable equilibrium point as time goes to infinity. Likewise, the ROA of a power system is viewed as a set of operating states such as rotor angles and frequencies able to converge to the stable equilibrium, which corresponds to the solution of power flow problem \cite{zz, Danwu,dhagash, Molzahn, aolaritei2017distributed, ali2017transversality}, after being subject to a disturbance. The corresponding convergent trajectory is regarded as a stable state trajectory. The estimation of ROA is basically dependent on the construction of Lyapunov function and its level set (see Fig. \ref{level}). In practice, constructing an analytic Lyapunov function for a nonlinear system is a challenging task.

The converse Lyapunov theorem \cite{kha96} allows for estimating the value of Lyapunov function without using its analytic form. For a power system that has a stable state trajectory $\mathbf{x}(t)$, $t \geq 0$, a commonly used converse Lyapunov function is $V(\mathbf{x})=\int_{0}^{\infty}\|\mathbf{x}(t)\|^2dt$ \cite{jone17}. We generalize converse Lyapunov function by introducing a more general function $\alpha(\cdot)$ (see the definition in Appendix \ref{app:def}) and a solution trajectory $\mathbf{\phi}(\cdot)$ in Lemma \ref{lya}. The existence of such generalized Lyapunov function is guaranteed in theory (i.e., Theorem $4.17$ in \cite{kha96}) as follows.
\begin{lemma}\label{lya}
Without loss of generality, let $\mathbf{x=0}$ be an asymptotically stable equilibrium point for the nonlinear system $\dot{\mathbf{x}}\mathbf{=f(x)}$, where $\mathbf{f}: X\rightarrow{R^n}$ is locally Lipschitz, and $S$ is the region of attraction, then there is a continuous positive definite function $W(\mathbf{x})$ such that
$$
V(\mathbf{x})=\int_{0}^{\infty}\alpha(\|\mathbf{\phi}(\mathbf{x},t)\|)dt, \quad V(\mathbf{0})=0
$$
and
\begin{equation} \label{eq:pror1}
 \frac{dV(\mathbf{x})}{dt}=\frac{\partial V(\mathbf{x})}{\partial \mathbf{x}}\mathbf{f}(\mathbf{x})\leq -W(\mathbf{x}), ~\forall~\mathbf{x} \in S
\end{equation}
with
\begin{equation} \label{eq:pror2}
\frac{d\mathbf{\phi}(\mathbf{x},t)}{dt}=\mathbf{f}(\mathbf{\phi}(\mathbf{x},t)), \quad\mathbf{\phi}(\mathbf{x},0)=\mathbf{x}.
\end{equation}
$\alpha(z)$ is a class $\Gamma$ function (see Appendix \ref{app:def}),
and the level set $\Omega_c=\{\mathbf{x}\in R^n~|~V(\mathbf{x})\leq c,~\forall~c>0 \}$ is a compact subset of $S=\{\mathbf{x}\in R^n |\lim_{t\rightarrow{+\infty}}\mathbf{\phi}(\mathbf{x},t)=\mathbf{0}\}$.
\end{lemma}

In Lemma \ref{lya}, the dynamics $\dot{\mathbf{x}} = \mathbf{f}(\mathbf{x})$ of a power system normally represent that of generators which are known as swing equations \cite{kun94}. Here we explicitly assume that the system dynamics are available, although our framework can be generally extended to incorporate uncertainties with a predefined complexity level, i.e., the smoothness of the unknown components, in the system's dynamical model. Property \eqref{eq:pror1} implies that the function $V(\mathbf{x})$ decays over time. Equation \eqref{eq:pror2} defines a stable trajectory $\mathbf{\phi}(\mathbf{x}, t)$ with the initial state $\mathbf{x}$. More importantly, the Lyapunov function $V(\mathbf{x})$ in Lemma \ref{lya} has a nice property that, by increasing the level set of $V(\mathbf{x})$, we can approach the real region of attraction (see Remark \ref{app:levelset} for a formal explanation). Note that this property may not hold for any analytical Lyapunov function. Thus, we can leverage the converse Lyapunov function approach to obtain a better ROA by collecting more sampling points in order to enlarge the level sets.
\begin{figure}\centering
\scalebox{0.05}[0.05]{\includegraphics{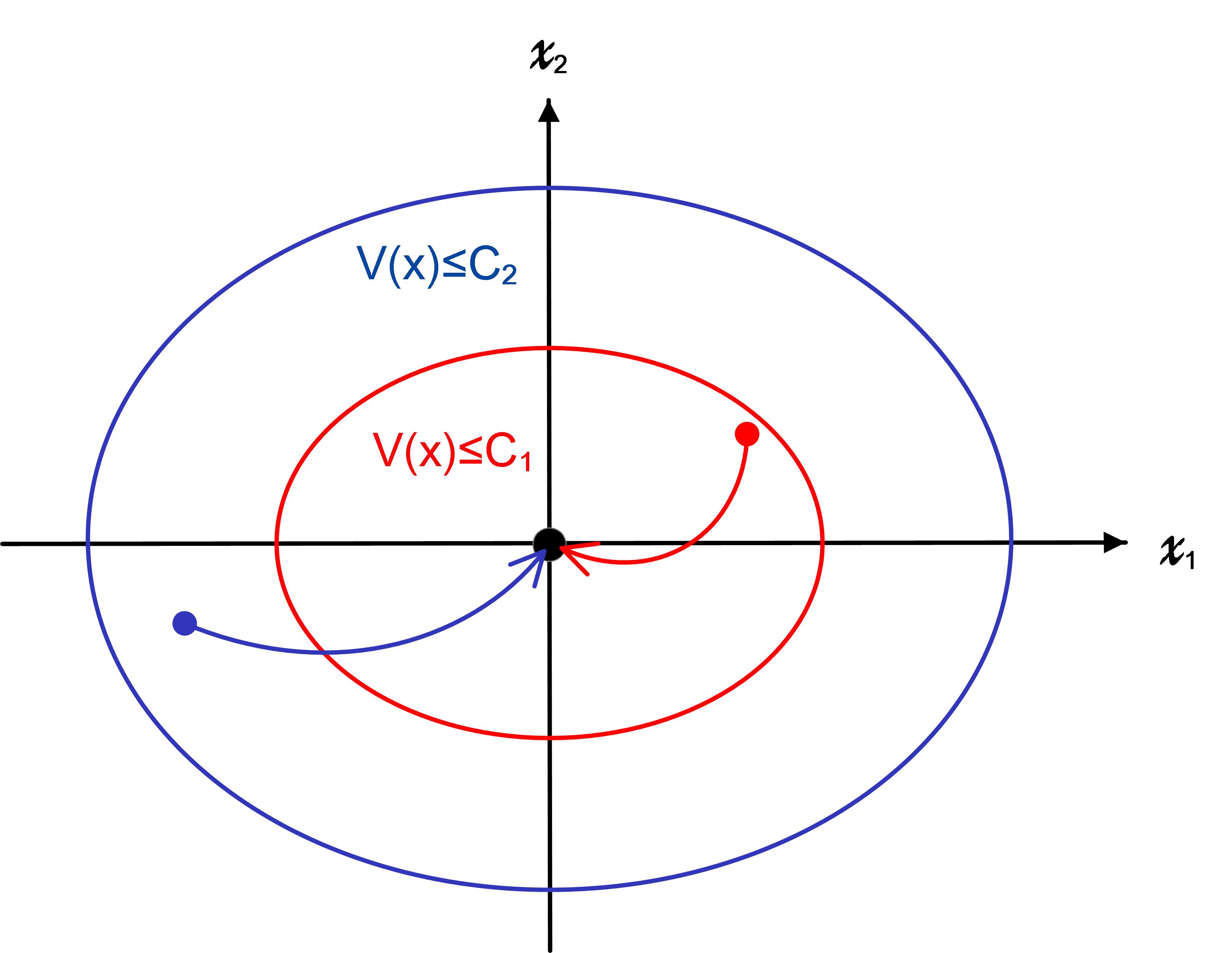}}
\caption{\label{level} Illustration on the level set of Lyapunov function. The red ellipse describes the level set $\{\mathbf{x}\in R^2~|~V(\mathbf{x})\leq C_1\}$, while the blue one denotes the level set $\{\mathbf{x}\in R^2~|~V(\mathbf{x})\leq C_2\}$ with $C_1<C_2$ and $\mathbf{x}=(x_1,x_2)$.
The black dot represents a stable equilibrium point of dynamical system. Each state in the level set can converge to the equilibrium point on condition of $\dot{V}(\mathbf{x})<0$.}
\end{figure}

Due to the discrete nature of sampling data, the converse Lyapunov function $V(\mathbf{x})$ in Lemma \ref{lya} can be discretized and approximated as follows
\begin{equation}\label{Vest}
    \hat{V}(\mathbf{x})=\sum_{i=1}^{n}\alpha(\|\mathbf{\phi}(\mathbf{x},t_i)\|)\Delta t,
\end{equation}
where $\Delta t$ denotes the sampling time interval
and $t_i=(i-1)\Delta t$, $i\in \{1,2,...,n\}$. The error caused by the above discretization can be further estimated (see Lemma \ref{lem:bound} in Appendix \ref{app:upbound}).

While the approximated $\hat{V}(\mathbf{x})$ can be calculated directly based on samples, the explicit Lyapunov function $V(\mathbf{x})$ is unknown. This work therefore learns this unknown Lyapunov function $V(\mathbf{x})$ from the discretized counterpart. By treating $V(\mathbf{x})$ and $\hat{V}(\mathbf{x})$ as a GP and its measurement, respectively, the estimation error $V(\mathbf{x})-\hat{V}(\mathbf{x})$ can be regarded as the measurement noise. This enables us to learn the unknown Lyapunov function $V(\mathbf{x})$ using the GP approach.

\section{GP for Learning Unknown Dynamics}\label{sec:gp}
The above section introduces the discretized, approximated Lyapunov function and the respective unknown Lyapunov function. In this section, we propose the GP approach for learning the unknown Lyapunov function (i.e., GP regression). Normally, a general GP regression requires a prior distribution of unknown functions specified by a mean function, a covariance function, and the probability of the observations and sampling data to obtain the posterior distribution. Without loss of generality, we consider an unknown function $h(\mathbf{x})$ as a GP, which can be sequentially measured by
\begin{equation*}
    y^{(i)} = h(\mathbf{x}^{(i)})+\epsilon, \quad i\in Z^{+}
\end{equation*}
where $y^{(i)}$ refers to the observed function value for the input $\mathbf{x}^{(i)}$ at the $i$-th sampling step, and the measurement noise $\epsilon$ is zero-mean, independent and bounded by $\sigma$. With the GP approach, we can obtain the posterior distribution over $h(\mathbf{x})$ by using sampling data in the training set $\{(\mathbf{x}^{(1)},y^{(1)}),(\mathbf{x}^{(2)},y^{(2)}),...,(\mathbf{x}^{(i)},y^{(i)})\}$.

In this work, it is assumed that the unknown Lyapunov function is a GP prior or its ``complexity" can be measured by the RKHS norm. The sampling data are obtained by implementing the GP-UCB based sampling algorithm. Note that, for an existing Lyapunov function, the GP approach allows to extend the certified ROA with a given confidence level.

\subsection{Gaussian process and RKHS norm}
\begin{figure}\centering
\scalebox{0.07}[0.07]{\includegraphics{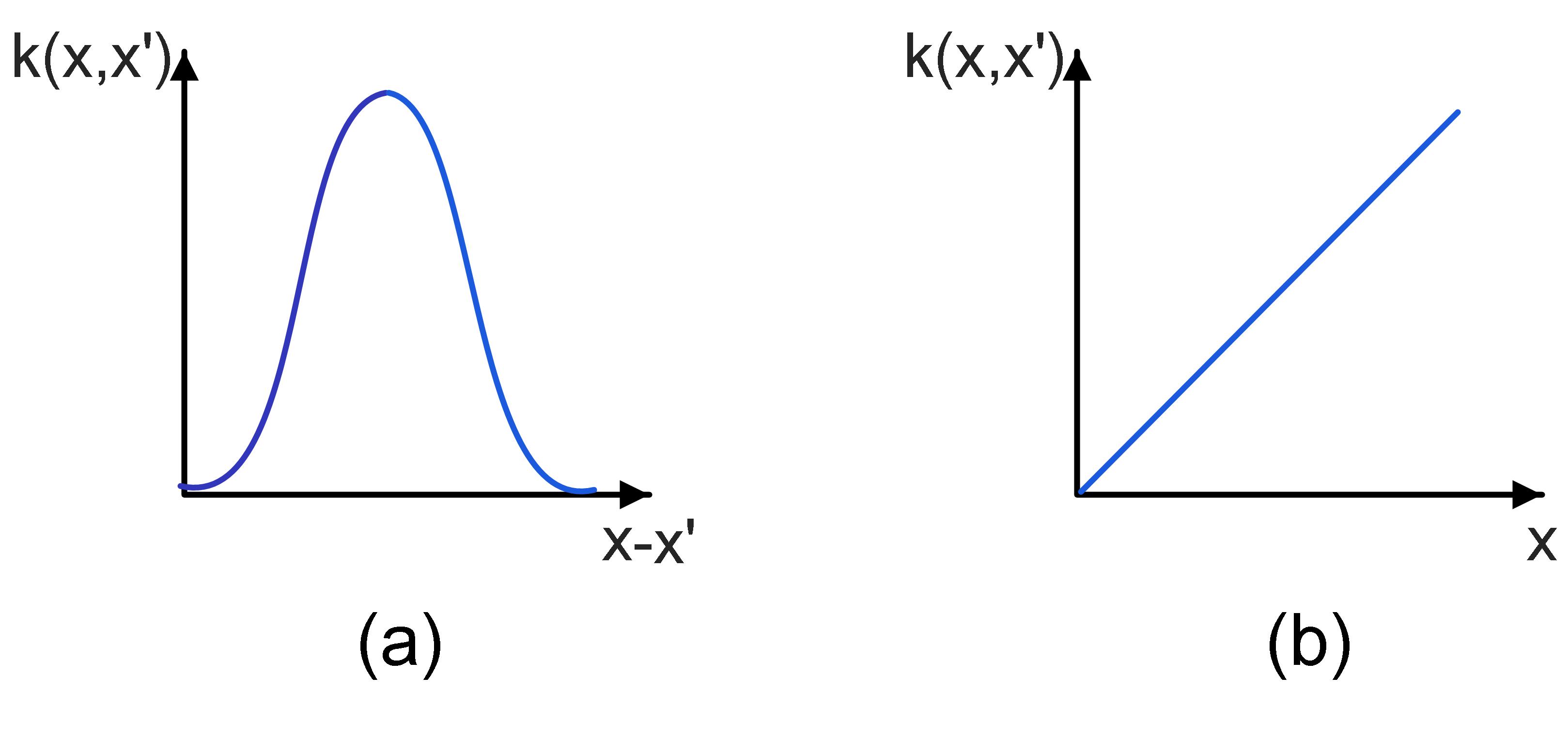}}
\caption{\label{kernel} Examples of basic kernel functions: (a) Squared Exponential Kernel $k(\mathbf{x},\mathbf{x}')=e^{-\|\mathbf{x}-\mathbf{x}'\|^2/2l^2}$ with a length scale parameter $l$ and (b) Linear Kernel $k(\mathbf{x},\mathbf{x}')=\mathbf{x}^T\mathbf{x}'$ with $\mathbf{x}'=1$. Each kernel allows to approximate the unknown function with a certain ``complexity".}
\end{figure}

By regarding the values of $V(\mathbf{x})$ as random variables, any finite collection of them is multivariate distributed in an overall consistent way. The unknown Lyapunov function $V(\mathbf{x})$ can be approximated by a GP. Note that the covariance or kernel function $k(\mathbf{x},\mathbf{x}')$ encodes the smoothness property of $V(\mathbf{x})$ from the GP (see Fig.~\ref{kernel}). For a sample from a known GP distribution $\hat{V}_N=[\hat{V}(\mathbf{x}^{(1)}),...,\hat{V}(\mathbf{x}^{(N)})]^T$ at points $\mathcal{A}_N=\{\mathbf{x}^{(1)},...,\mathbf{x}^{(N)}\}$, $\hat{V}(\mathbf{x}^{(i)})=V(\mathbf{x}^{(i)})+\epsilon$ with $\epsilon\sim N(0,\sigma^2)$, there are the analytic formulas for mean $\mu_N(\mathbf{x})$, covariance $k_N(\mathbf{x},\mathbf{x}')$ and variance $\sigma_N^2(\mathbf{x})$ of the posterior distribution as follows \cite{sri12}
\begin{equation}\label{mu_sig}
    \begin{split}
        \mu_N(\mathbf{x})&=k_N(\mathbf{x})^T(K_N+\sigma^2I)^{-1}\hat{V}_N \\
       k_N(\mathbf{x},\mathbf{x}')&=k(\mathbf{x},\mathbf{x}')-k_N(\mathbf{x})^T(K_N+\sigma^2I)^{-1}k_N(\mathbf{x}') \\
       \sigma_N^2(\mathbf{x})&=k_N(\mathbf{x},\mathbf{x})
    \end{split}
\end{equation}
where $k_N(\mathbf{x})=[k(\mathbf{x}^{(1)},\mathbf{x}),...,k(\mathbf{x}^{(N)},\mathbf{x})]^T$ and $K_N$ is the positive definite kernel matrix $[k(\mathbf{x},\mathbf{x}')]_{\mathbf{x},\mathbf{x}'\in \mathcal{A}_N}$.

If the prior distribution over $V(\mathbf{x})$ is unknown, $V(\mathbf{x})$ can be approximated by the linear combination of kernel functions in the RKHS, i.e. $V(\mathbf{x})=\sum_{i}c_i k(\mathbf{x},\mathbf{x}^{(i)})$.
The RKHS $H_k(X)$ with the kernel $k(\mathbf{x},\mathbf{x}')$ is a complete subspace of $L_2(X)$, and its inner product $\langle\cdot,\cdot\rangle_k$ is endowed with the reproducing property: $\langle V(\mathbf{x}),k(\mathbf{x},\cdot)\rangle_k=V(\mathbf{x})$, $\forall~V(\mathbf{x})\in H_k(X)$. Moreover, there exists a uniquely associated RKHS $H_k(X)$ for each kernel $k(\mathbf{x},\mathbf{x}')$ \cite{kan18}. The induced RKHS norm $\|V\|_k=\sqrt{\langle V,V\rangle_k}$ is used to quantify the smoothness of the function $V$ with respect to the kernel $k(\mathbf{x},\mathbf{x}')$. In brief, the function $V(\mathbf{x})$ gets smoother as $\|V\|_k$ decreases. In this work, we consider that Lyapunov functions $V(\mathbf{x})$ have relatively low ``complexity" or high smoothness, which can be measured by the RKHS norm.

\begin{remark}
There is an interesting connection between the kernel function in RKHS and the covariance function of a GP. If the kernel function in RKHS is the same as the covariance function of a GP, the posterior mean of a GP is equivalent to the estimator of kernel ridge regression in RKHS \cite{kim70}. In addition, the posterior variance of a GP can be regarded as a worst case error in RKHS.
\end{remark}

Within the GP approach, how one samples plays an important role because the sampling rule affects the confidence level of the estimated ROA, which is the probability at least with which a trajectory initiated from an inner state of such estimated ROA will converge to the corresponding stable equilibrium. In other words, that is the probability that an estimated ROA is valid. For this work, we use the GP-UCB based sampling rule. Compared with other heuristics in GP optimization \cite{moc75,moc12,bro10}, the GP-UCB based sampling rule is able to deal with the trade-off between exploitation for optimizing the objective function and exploration for reducing the uncertainty with the guaranteed theoretical performance.

\subsection{GP-UCB based algorithm}
For a given value of $\delta \in (0,1)$ and the sampling domain $X\in R^n$, which is a subset of the state space, our goal is to maximize the region of attraction, wherein each point converges to the origin with probability at least the confident level of $1-\delta$. Thus, a GP-UCB based algorithm is developed in Table \ref{tab:gpa} to select the sampling points for enlarging the ROA with a guaranteed confidence level. Specifically, a sampling point $\mathbf{x}^{(i)}$ is selected in $X$ by searching for the maxima of $\mu_{i-1}(\mathbf{x})+\beta^{1/2}_{i} \sigma_{i-1}(\mathbf{x})$, where the term $\mu_{i-1}(\mathbf{x})$ contributes to enlarge the level set of Lyapunov function and the term $\sigma_{i-1}(\mathbf{x})$ allows to reduce the uncertainty of sampling region.

Essentially, the sampling rule aims to reconcile the trade-off between the exploitation for enlarging the ROA and the exploration for reducing the uncertainty of sampling region. Let the sampling point $\mathbf{x}^{(i)}$ serve as the initial condition of nonlinear system $\dot{\mathbf{x}}=\mathbf{f}(\mathbf{x})$, and this allows to generate a state trajectory $\mathbf{\phi}(\mathbf{x}^{(i)},t)$, $t\geq0$. If this state trajectory can converge to the origin, the point $\mathbf{x}^{(i)}$ is called as a stable sampling point. Then the value of Lyapunov function at $\mathbf{x}^{(i)}$ is estimated by $\hat{V}(\mathbf{x}^{(i)})$ with (\ref{Vest}). By choosing $\{(\mathbf{x}^{(1)},\hat{V}(\mathbf{x}^{(1)})),...,(\mathbf{x}^{(i)},\hat{V}(\mathbf{x}^{(i)}))\}$ as the training set, $\mu_i(\mathbf{x})$ and $\sigma_i(\mathbf{x})$ for the unknown Lyapunov function $V(\mathbf{x})$ can be updated according to (\ref{mu_sig}). On the other hand, if the state trajectory $\mathbf{\phi}(\mathbf{x}^{(i)},t)$, $t\geq0$ fails to converge to the origin, the sampling point $\mathbf{x}^{(i)}$ is removed from the sampling region $X$. The above process does not terminate until it achieves the specified number of stable sampling points.

\begin{table}
 \caption{\label{tab:gpa} GP-UCB based Algorithm.}
 \begin{center}
 \begin{tabular}{lcl} \hline
 \textbf{Input:}  $X\in R^n$, $\delta$, $\xi$, $t_n$, $\mu_0$, $\sigma_0$, $k(\mathbf{x},\mathbf{x}')$, $i=1$ \\
 \textbf{Output:} $\mathbf{x}^{(i)}$, $\mathbf{\phi}(\mathbf{x}^{(i)},t)$, $\hat{V}(\mathbf{x}^{(i)})$, $\mu_i(\mathbf{x})$, $\sigma_i(\mathbf{x})$, $i\in\{1,2,...,N\}$ \\ \hline
  1: ~~\textbf{while} $(i\leq N)$ \\
  2:~~~~~~~ Choose $\mathbf{x}^{(i)}=\text{argmax}_{\mathbf{x}\in X}\left[\mu_{i-1}(\mathbf{x})+\beta^{1/2}_{i} \sigma_{i-1}(\mathbf{x})\right]$ \\
  3:~~~~~~~ Generate a state trajectory $\mathbf{\phi}(\mathbf{x}^{(i)},t)$, $t\geq0$ \\
  4:~~~~~~~ \textbf{if} ($\|\mathbf{\phi}(\mathbf{x}^{(i)},t_n)\|<\xi$) \\
  5:~~~~~~~~~~~ Sample $\hat{V}(\mathbf{x}^{(i)})=V(\mathbf{x}^{(i)})+\epsilon$ with (\ref{Vest}) \\
  6:~~~~~~~~~~~ Update $\mu_{i}(\mathbf{x})$ and $\sigma_{i}(\mathbf{x})$ with (\ref{mu_sig}) \\
  7:~~~~~~~~~~~ Update $i=i+1$ \\
  8:~~~~~~~ \textbf{else} \\
  9:~~~~~~~~~~~ Update $X=X/\{\mathbf{x}^{(i)}\}$. \\
  10:~~~~~~ \textbf{end if} \\
  11: ~\textbf{end while} \\ \hline
 \end{tabular}
 \end{center}
\end{table}

\begin{figure}\centering
 {\includegraphics[width=0.45\textwidth]{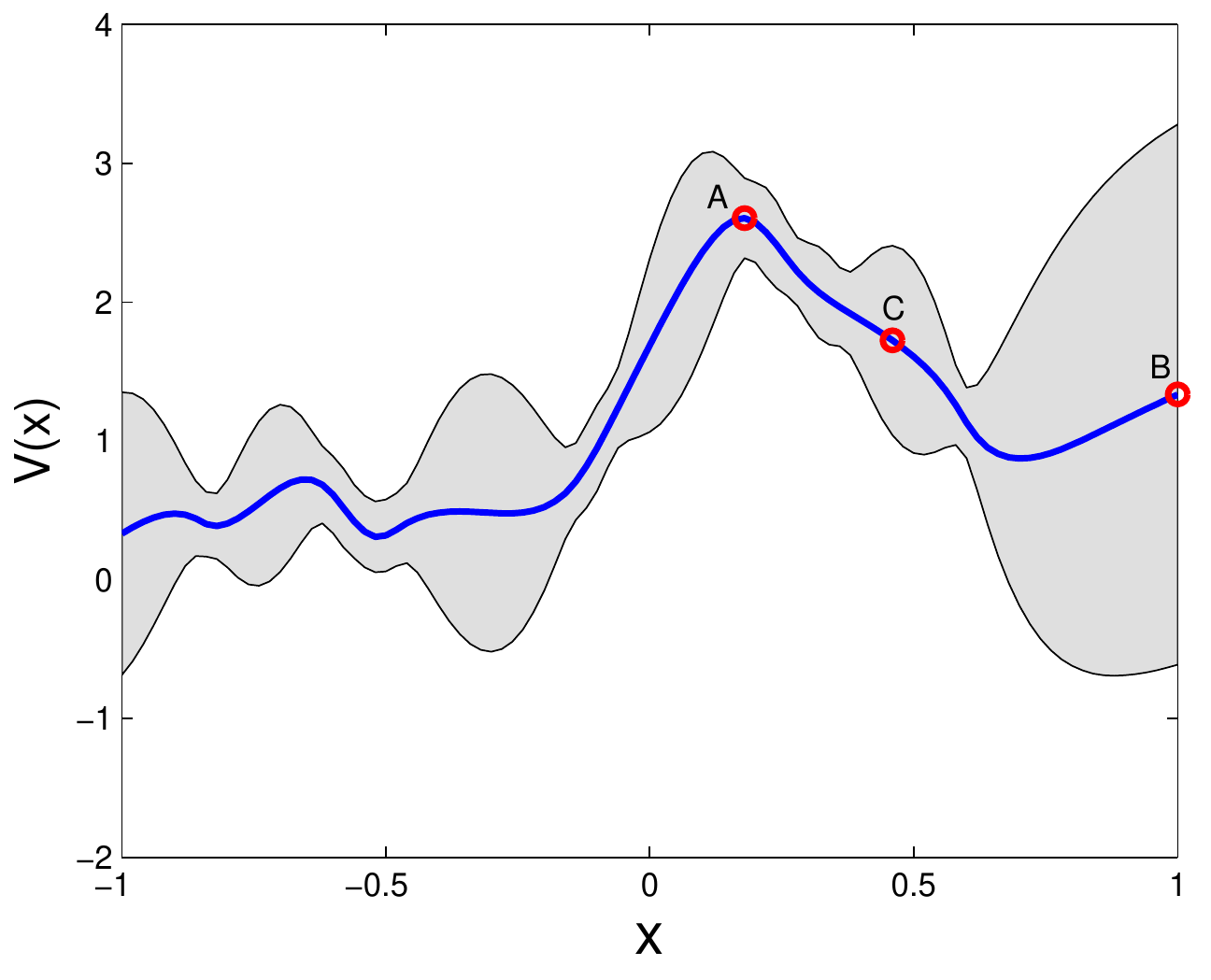}}
 \caption{\label{gproa} Illustration on three different sampling rules. The blue line denotes mean values of the unknown function $V(\mathbf{x})$, and the gray shade describes the $95\%$ confidence interval. The three red circles $A$, $B$ and $C$ represent three different sampling rules, respectively. $A$ prefers the point with the maximum posterior mean value in order to achieve the large $V(\mathbf{x})$, and $B$ selects the point with the maximum posterior variance in order to reduce the uncertainty. $C$ aims to allow for both the exploitation for large $V(\mathbf{x})$ and the exploration for eliminating the uncertainty.}
\end{figure}

Figure \ref{gproa} illustrates three different sampling schemes, i.e., Scheme A, Scheme B, Scheme C which correspond to sampling points A, B, C, respectively. Mathematically, Scheme $A$ can be described as $\mathbf{x}^{(i)}=\text{argmax}_{\mathbf{x}\in X}\mu_{i-1}(\mathbf{x})$ and Scheme $B$ is expressed as $\mathbf{x}^{(i)}=\text{argmax}_{\mathbf{x}\in X}\sigma_{i-1}(\mathbf{x})$.
In practice, Scheme $A$ is too greedy and tends to get local optima, while Scheme $B$ provides a good rule for exploring $V(\mathbf{x})$ globally \cite{sri12}. As a result, Scheme $C$ is designed to integrate $A$ with $B$ by adopting $\mathbf{x}^{(i)}=\text{argmax}_{\mathbf{x}\in X}J_i(\mathbf{x})$, where the reward function $J_i(\mathbf{x})$ is given by
$$
J_i(\mathbf{x})=\mu_{i-1}(\mathbf{x})+\beta^{1/2}_i\sigma_{i-1}(\mathbf{x}).
$$
In addition, the parameter $\beta_i$ depends on the ``complexity" of Lyapunov function, the sample size and information gain (see Appendix \ref{app:gain}).

\begin{remark}
To obtain the global maximum of $J_i(\mathbf{x})$ in the sampling region $X$ is generally infeasible due to the non-convexity of $J_i(\mathbf{x})$. In practice, there are multiple ways to heuristically search for its local maxima. For instance, the local maxima can be readily identified from the finite historic dataset of sampling points instead of the region $X$. In addition, it is expected to approach the local maxima along the gradient ascent of $J_i(\mathbf{x})$, and the gradient can be approximated with the finite difference method \cite{sto13}.
\end{remark}


\begin{remark}
For a certified ROA, it is sufficient to judge a stable
sampling point if the corresponding state trajectory can overpass the boundary of this ROA without checking its convergence to the stable equilibrium point. This will reduce the computation time for selecting the stable sampling points.
\end{remark}

\section{Main Results}\label{sec:main}
This section presents main theoretical results on the construction and evaluation of ROA with a given confidence level by using the GP-UCB based algorithm.

\begin{theorem}\label{betaT}
Let $\delta\in(0,1)$, and the measurement noise is bounded by $\sigma$. Then it holds with probability at least $1-\delta$ that
\begin{equation*}
    |V(\mathbf{x})-\mu_{N-1}(\mathbf{x})|\leq \beta^{1/2}_N \sigma_{N-1}(\mathbf{x}), \quad \forall \mathbf{x}\in X, ~~\forall N\in Z^{+}
\end{equation*}
where $\beta_N=2\|V\|_k^2+300\gamma_N\ln^3{(N/\delta)}$.
\end{theorem}
\begin{proof}
The results follow directly from Theorem 6 in \cite{sri12}.
\end{proof}

\begin{remark}
The RKHS norm $\|\cdot\|_k$ characterizes the ``complexity" or ``smoothness" of the function in RKHS. If the bound of $\|V\|^2_k$ is unknown a prior, it can be computed using kernel ridge regression as $\|V\|^2_k=\hat{V}^T_NP^T\text{diag}(\lambda_i(\lambda_i+N\theta)^{-2})P\hat{V}_N$,
where $K_N=P^T\text{diag}(\lambda_i)P$ and $\text{diag}(\lambda_i)\in R^{N\times N}$ denotes a diagonal matrix with the diagonal elements $\lambda_i$, $i=1,2,...N$. Moreover, $\theta$ is a positive tunable parameter for controlling the smoothness of $V(\mathbf{x})$ to avoid overfitting. A larger value of $\theta$ leads to the smoother function $V(\mathbf{x})$ (see Appendix \ref{app:RKHS}).
\end{remark}

For a stable equilibrium point, the Lyapunov function $V(\mathbf{x})$ is viewed as a GP, and Theorem \ref{betaT} allows to estimate the ROA of nonlinear system with a given confidence level. If the ROA $\Omega^\star$ is already established according to an existing Lyapunov function $V^\star(\mathbf{x})$, the proposed approach is applied to evaluate the stability of the region outside the ROA $\Omega^\star$ by treating the mismatch between $V(\mathbf{x})$ and $V^\star(\mathbf{x})$ as a GP. Thus, theoretical results are summarized as follows.

\begin{theorem}\label{theo}
Let $\delta\in(0,1)$ and $V^\star(\mathbf{x})$ be an existing Lyapunov function for nonlinear system $\dot{\mathbf{x}}=\mathbf{f(x)}$. With GP-UCB based Algorithm in Table \ref{tab:gpa},
it holds that
$$
\text{Prob}(\mathbf{x}\in S)\geq 1-\delta, \quad \forall \mathbf{x} \in S_{\delta,N}
$$
where $S_{\delta,N}$ is given by
$$
\left\{\mathbf{x}\in X ~|~ V^\star(\mathbf{x})+\mu_{N-1}(\mathbf{x})+\bar{\beta}^{1/2}_N \sigma_{N-1}(\mathbf{x})\leq C_{\max,N}\right\}
$$
with $C_{\max,N}=\max_{\mathbf{x}^{(i)}\in \mathcal{A}_N}\hat{V}(\mathbf{x}^{(i)})$ and
$$
\bar{\beta}_N=2\left\|V(\mathbf{x})-V^\star(\mathbf{x})\right\|_k^2+300\gamma_N\ln^3{(N/\delta)}.
$$
For a stable equilibrium point, the above conclusion also holds with $V^\star(\mathbf{x})\equiv0$.
\end{theorem}

\begin{proof}
See Appendix \ref{app:theo}.
\end{proof}


\begin{remark}
There are alternative schemes to characterize the mismatch between the converse Lyapunov function $V(\mathbf{x})$ and the existing Lyapunov function $V^\star(\mathbf{x})$ other than the difference scheme $\Delta V(\mathbf{x})=V(\mathbf{x})-V^\star(\mathbf{x})$. For example, it is also feasible to adopt the proportion scheme $V(\mathbf{x})/V^\star(\mathbf{x})$ as an unknown function for the GP learning. Essentially, a stable equilibrium point can be thought of as a special case of the ROA when the ROA shrinks into a point and thus the existing Lyapunov function becomes $V^\star(\mathbf{x})\equiv0$.
\end{remark}



\section{Numerical Simulations}\label{sec:sim}
This section presents simulation results using the GP-UCB based algorithm for both single machine infinite bus (SMIB) system and IEEE 39 bus system. First of all, swing dynamics of a power system are introduced as follows.

\subsection{Power system model}
Consider a $M$-bus power system described by a set of swing equations as follows \cite{mun13,sau98}
\begin{equation} \label{eq:swing}
    m_i\Ddot{\psi}_i+d_i\dot{\psi}_i=p_i-\sum_{j\in \mathcal{N}_{i}}\frac{1}{B_{ij}}\sin(\theta^\star_{ij}+\psi_{ij}),
\end{equation}
where $\psi_i=\theta_i-\theta^\star_{i}$, $\psi_{ij}=\psi_i-\psi_j$ and $\theta^\star_{ij}=\theta^\star_{i}-\theta^\star_{j}$. Note that $\theta_{i}$ denotes the angle of the generator at bus $i$, and the superscript $\star$ represents the steady state condition or the power flow solution. The swing equation \eqref{eq:swing} describes the evolution of the phase angle due to the power mismatch between the mechanical power $p_i$ and the electrical power. In addition, $\mathcal{N}_{i}$ is the neighbourhood set of bus $i$, and $B_{ij}$ denotes the susceptance of branch $ij$. The parameters $m_i$ and $d_i$ refer to inertia and damping coefficients for the machine, respectively. For a load bus $i$, we normally assume that $m_i=d_i=0$.  In the steady state condition, the mechanical power $p_i$ can be expressed as $p_i=\sum_{j\in\mathcal{N}_i} \,1 / B_{ij} \sin{\theta^\star_{ij}}$.
This allows us to obtain the following perturbed dynamic model
\begin{equation*}
    m_i\Ddot{\psi}_i+d_i\dot{\psi}_i=\sum_{j\in \mathcal{N}_{i}}\frac{1}{B_{ij}}\left[\sin{\theta^\star_{ij}}-\sin(\theta^\star_{ij}+\psi_{ij})\right].
\end{equation*}
A number of approaches to assess the stability of these systems are presented in \cite{pareek, dongchan,hungrobust,hungmulti, tuyen,tuyen2017analysis,Petr}. Here, we use an energy-like function for the above system of nonlinear differential equations as follows
\begin{equation*}
    \begin{split}
        V^\star(\dot{\mathbf{\psi}},\mathbf{\psi}) &=\frac{1}{2}\sum_{i=1}^{M}\sum_{j\in \mathcal{N}_i}\frac{1}{B_{ij}}\int_{0}^{\psi_{ij}}\left[\sin{\theta^\star_{ij}}-\sin(\theta^\star_{ij}+\tau)\right]d\tau \\
        &~+\frac{1}{2}\sum_{i=1}^{M}m_i\dot{\psi}^2_i
    \end{split}
\end{equation*}
and the origin (i.e., $\mathbf{\psi}=\dot{\mathbf{\psi}}=0$) is locally asymptotically stable, and the certified ROA can be estimated by \cite{mun13}
\begin{equation}\label{roa_est}
    \Omega^\star=\left\{(\mathbf{\psi},\dot{\mathbf{\psi}})~|~\mathbf{\psi}^TL\mathbf{\psi}+\dot{\mathbf{\psi}}^T\Lambda\dot{\mathbf{\psi}}\leq C_{\lambda}\right\},
\end{equation}
where $\mathbf{\psi}=(\psi_1,...,\psi_M)^T$ and $\dot{\mathbf{\psi}}=(\dot{\psi}_1,...,\dot{\psi}_M)^T$. In addition, $L$ denotes the Laplacian matrix of power network with weights $1/B_{ij}$, and $\Lambda$ is a diagonal matrix that satisfies $\Lambda=\text{diag}(\mathbf{m})$ with $\mathbf{m}=(m_1,...,m_M)^T$. Moreover, $C_{\lambda}$ is given by
$C_{\lambda}=\min_{ij}1/B_{ij}\left(2\cos{\lambda}-(\pi-2\lambda)\sin{\lambda}\right)$ with $\lambda=\max_{ij}|\theta^\star_{ij}|$.
For simplicity, the class $\Gamma$ function $\alpha(z)$ in Definition \ref{def} is chosen as $\alpha(z)=z^2$ in order to estimate the value of converse Lyapunov function.

\subsection{Single machine infinite bus system}
\begin{figure}\centering
 {\includegraphics[width=0.42\textwidth]{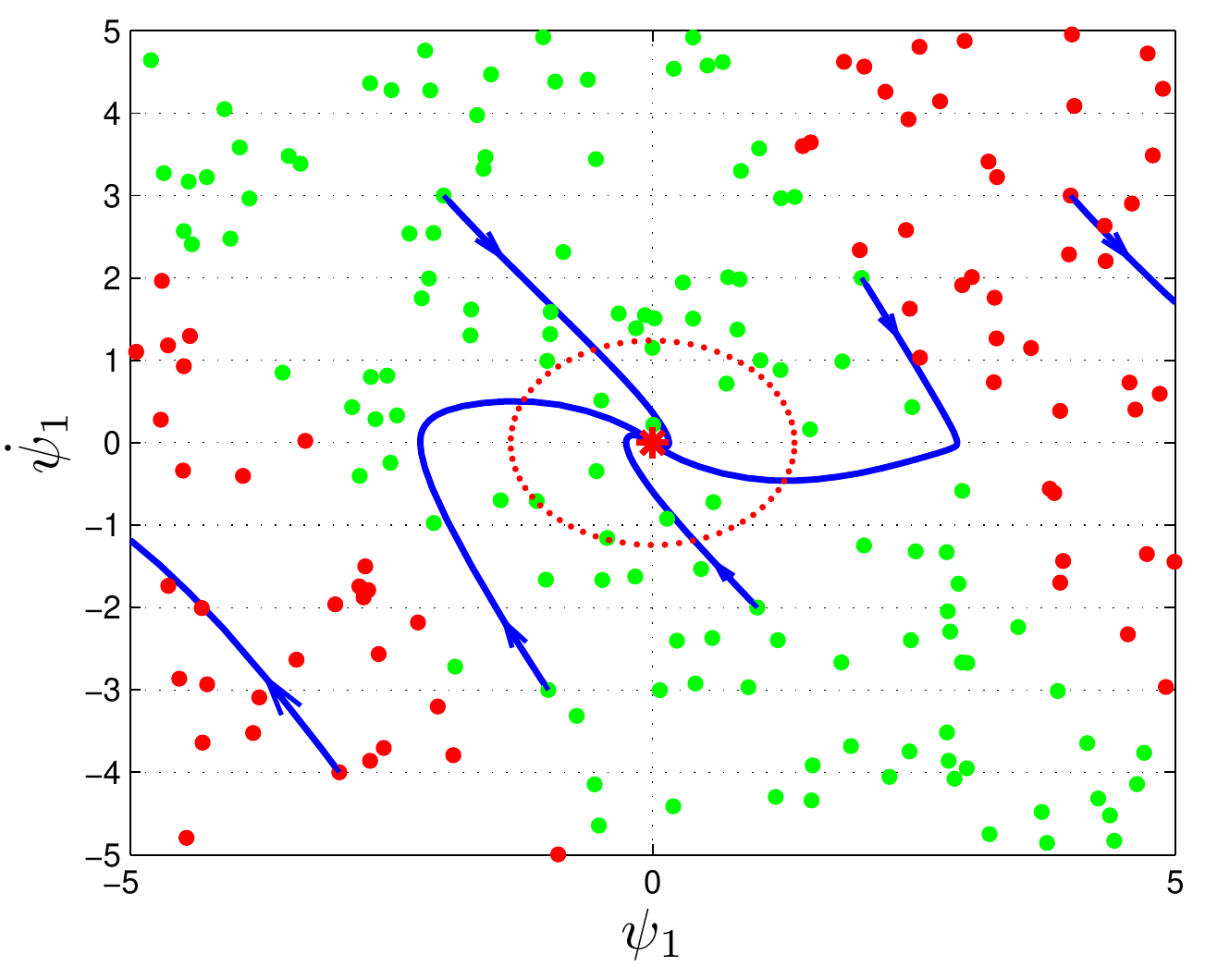}}
 \caption{\label{traj} Sampling points and state trajectories of SMIB system. Red dots represents the unstable sampling points that fail to converge to the origin, while green dots refer to stable sampling points that converge to the origin. Blue lines indicate their state trajectories, and the arrows point in the direction of state trajectories. The red dashed ellipse denotes the certified ROA according to (\ref{roa_est}).}
\end{figure}

In this SMIB system, bus $1$ is the generator bus which is connected to the infinity bus $2$ where $\theta_2 = \psi_2 = 0$. The swing equation for this SMIB system is given by
\begin{equation*}
    m_1\Ddot{\psi}_1+d_1\dot{\psi}_1=\frac{1}{B_{12}}\left[\sin{\theta^\star_{1}}-\sin(\theta^\star_{1}+\psi_{1})\right]
\end{equation*}
where $m_1=12$, $d_1=20$, $p_1=0.5$ and $B_{12}=0.1$, and thus we have $\sin\left(\theta^\star_{12}\right) = 0.05$ and $\theta^\star_{12}=\arcsin(0.05)$. All values are in p.u. Then the ROA can be estimated by $\{(\psi_1,\dot{\psi}_1)|10\psi^2_1+12\dot{\psi}^2_1\leq 18.45\}$ according to (\ref{roa_est}) and it is described by the red dashed ellipse in Fig. \ref{traj}. The GP-UCB based Algorithm in Table \ref{tab:gpa} is adopted to assess the confidence level of operating states around the certified ROA. The parameter setting is given as follows: $N=100$, $\delta=0.05$, $\xi=0.01$, $\Delta t=0.01$ and $t_n=100$. In addition, the squared exponential kernel or covariance function is employed to learn the unknown Lyapunov function with the unit characteristic length-scale and $\mu_0=0$. Figure \ref{roa100} presents the simulation result for the SMIB system. The sampling points are denoted by small green dots, and the region of operating states with confidence level at least $95\%$ has been marked in yellow. This implies that each state point in the yellow region converges to the origin with the probability that is not less than $95\%$. It is observed that the red dashed ellipse is closely surrounded by the yellow region. Notably, the yellow region covers most sampling points outside the certified ROA except for four sampling points that are far away from the origin.

\begin{figure}\centering
 {\includegraphics[width=0.435\textwidth]{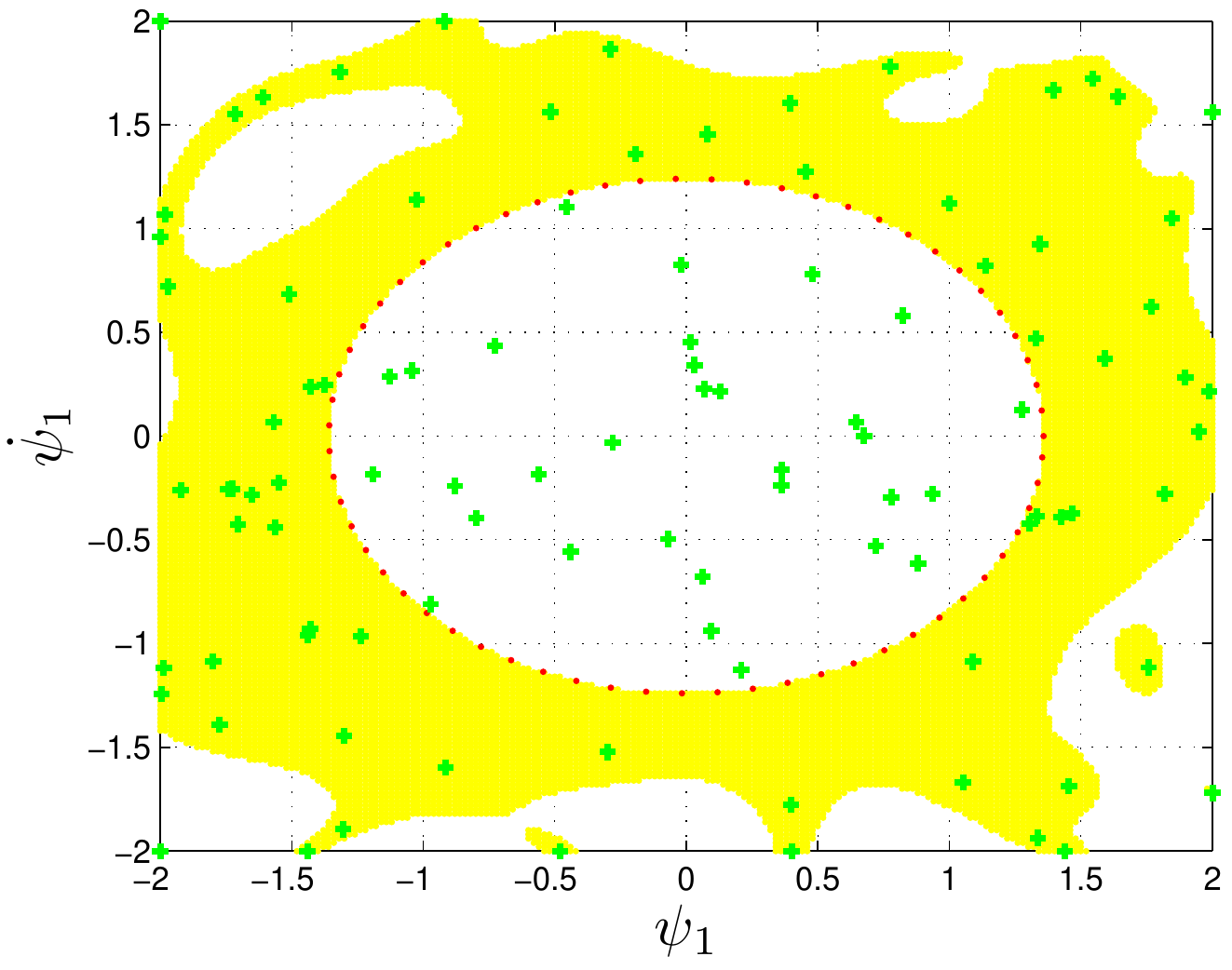}}
 \caption{\label{roa100} Confidence evaluation for the ROA of SMIB system with $100$ sampling points. The green dots denote the stable sampling points, and red ellipse region refers to the certified ROA with an existing Lyapunov function. The yellow region indicates the state that converges to the origin with the probability at least $95\%$.}
\end{figure}

\subsection{IEEE 39 bus system}

In order to validate the scalability of the proposed approach, numerical simulations are also conducted on  the IEEE 39-bus 10-machine system (see Fig. \ref{b39}), and the parameters are the same as those in \cite{dor14}. Bus $31$ with a generator is assigned as the swing bus. We then consider the dynamics for the remaining $9$ machines. The detailed parameter setting and the algorithm can be found in our MATLAB code made available in GitHub \cite{github}. Figure \ref{roa39} shows the assessment result for the ROA of IEEE $39$ bus system with $9$ machines. To facilitate the visualization, we project the stable sampling points, the certified ROA and the confidence region onto $9$ distinct two-dimensional planes, respectively. Essentially, each plane acts as a cross section to showcase the profile of confidence region and certified ROA with respect to a different machine. In each panel of Fig. \ref{roa39}, the state points inside the red dashed ellipse is guaranteed to converge to the origin, while those in the yellow region are asymptotically convergent with the probability at least $95\%$.

By Monte Carlo like estimate, the volume of yellow region is about $2.3\times10^4$ times larger than that of the certified ROA. This achievement is due to the large number of state dimensions or the size of the considered dynamical system. Note that this comparison is not entirely fair as the certified ROA can ensure that the system state will always converge to the stable equilibrium if it starts from the inside, while our estimated yellow region is only $95\%$ confident.

Similar to the case of the SMIB system, the yellow region does not cover the state points that are relatively far away from the origin. Actually, the profile of yellow regions largely depends on the distribution and size of sampling points as well as the choice of kernel functions for GP learning.

\begin{figure}\centering
\scalebox{0.1}[0.1]{\includegraphics{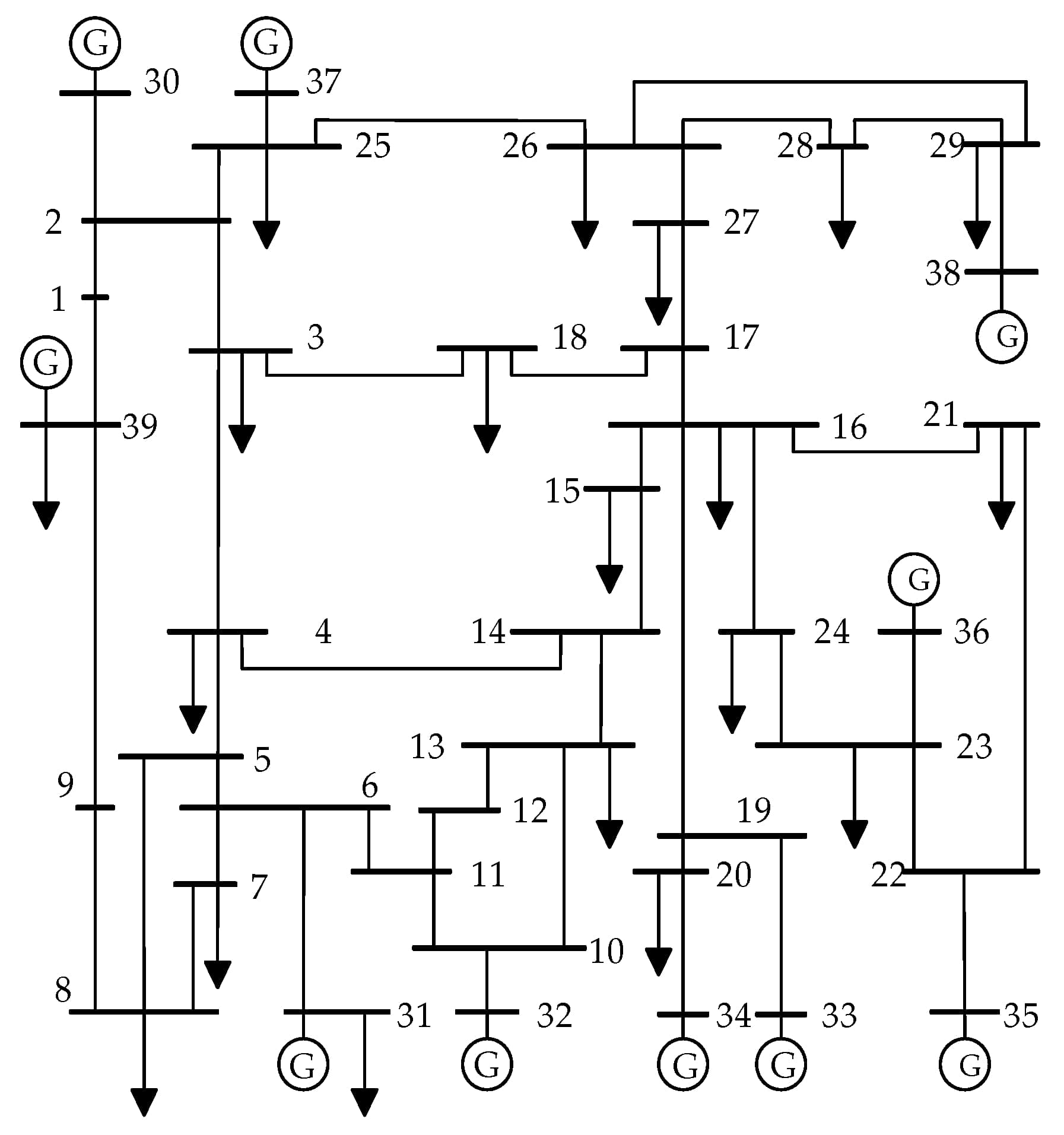}}
\caption{\label{b39} IEEE 39 Bus System}
\end{figure}

\begin{figure}\centering
 {\includegraphics[width=0.48\textwidth]{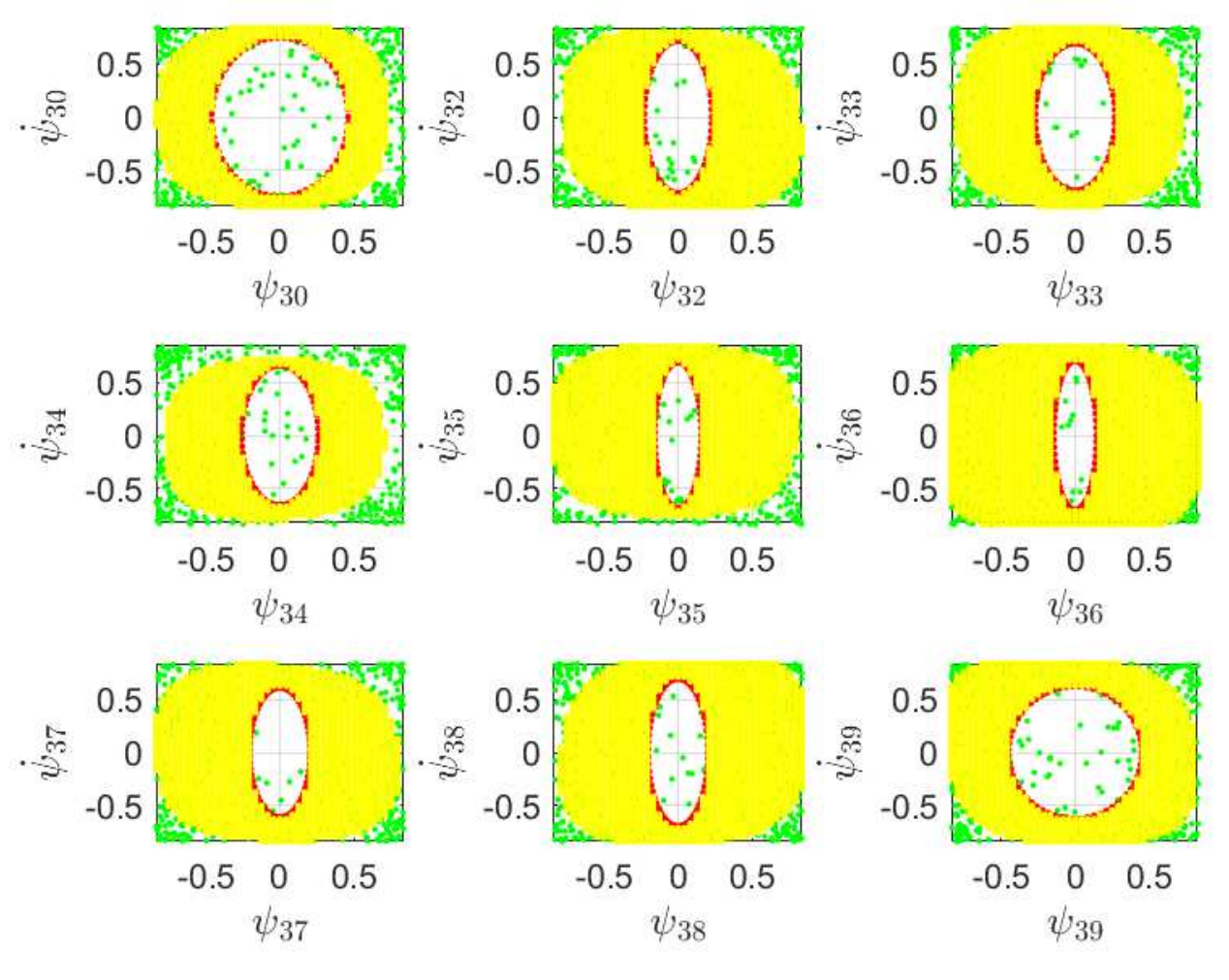}}\caption{\label{roa39} Confidence evaluation for the ROA of IEEE 39 bus system with $400$ sampling points. In each panel, the green dots denote the stable sampling points, and the red dashed ellipse represents the boundary of a certified ROA. The yellow region indicates the state that converges to the origin with the probability at least $95\%$.}
\end{figure}

\subsection{Discussions on computational cost}
The computational burden associated with GP-UCB algorithm is mostly due to two processes: solving the dynamical equation or swing equation using ODE solvers, and GP optimization presented in \eqref{mu_sig}. While the former depends heavily on the system size, the dimension of differential equation for power systems does not result in the visible increase of computational cost in GP optimization. Indeed, the computational cost is essentially immune to the dimension of power system dynamics. This is because the computation burden largely results from the operation of matrix inverse in (\ref{mu_sig}), where the size of kernel matrix $K_N$  mainly depends on the number of sampling points $N$ rather than the dimension of power system dynamics. With Matlab 2017b in the desktop (Intel i7-3770 CPU 3.40GHz and installed RAM 8GB), it takes around $10$ minutes for IEEE $39$-bus system and $50$ seconds for the SMIB system in our simulations. For the large-size sampling data, many efficient approaches for training the GP with the desirable performance \cite{liu18} are available. For instance, the sparse representation of GP model is developed to overcome the limitations for large data sets via the sequential construction of a sub-sample of the entire sampling data \cite{csa02}.

\section{Conclusions and Future Work}\label{sec:con}
In this paper, we investigated the problem of estimating the ROA for power systems. By treating the unknown Lyapunov function as a Gaussian Process, we assessed state stability
of power systems with the aid of the converse Lyapunov function. For an existing Lyapunov function, our approach allows for extending the pre-existing ROA with a provable confidence level. In addition, a GP-UCB based algorithm was developed to deal with the trade-off between exploration and exploitation in selecting stable sampling points. Numerical simulations are conducted to validate the proposed approach on the IEEE test cases.

In the next step, we will consider the online learning applications of unknown dynamics in practical power systems and the real-time prediction and stability assessment. This requires the creation of an efficient numerical algorithm with the aid of Gaussian Process and the converse Lyapunov function, which introduces the second part of this work. Another improvement is to optimally rescale state variables to be aligned to the shape of the real ROA which is not equal in all dimensions.

\section*{Acknowledgement}
The work of Chao Zhai and Hung Nguyen is supported by NTU SUG.

\section*{Appendix}
This section provides a mathematical definition of the class $\Gamma$ function and theoretical proofs for Lemma \ref{lem:bound}, Lemma \ref{gamma} and Theorem \ref{theo}, respectively. First of all, the definition is presented as follows.

\subsection{The class $\Gamma$ function}\label{app:def}

\begin{definition}\label{def}
The class $\Gamma$ function consists of all continuous functions $\alpha: [0,a)\rightarrow{[0,\infty]}$ which satisfy the following conditions:
\begin{enumerate}
  \item $\forall z>0$, $\alpha(z)\in C^2$.
  \item $\forall z>y\geq0$, $\alpha(z)>\alpha(y)$ and $\alpha(0)=0$.
  \item $\forall z\geq0$, $\exists~m>0$, such that $\alpha(z)\leq z^m$.
\end{enumerate}
\end{definition}

\begin{remark}
Condition $1$ in Definition \ref{def} indicates that the first two derivatives of the class $\Gamma$ function $\alpha(z)$ exist and are continuous. Condition $2$ implies that $\alpha(z)$ is a strictly increasing function. In addition, Condition $3$ aims to impose the restriction on the rate of $\alpha(z)$.
\end{remark}

\begin{remark} \label{app:levelset}
The construction of $V(\mathbf{x})$ in Lemma \ref{lya} allows us to obtain an important property on the region of attraction $S$ (i.e., $\lim_{c\rightarrow{+\infty}}\Omega_c=S$). Specifically, $\forall \mathbf{x}\in \Omega_c$, we have $\int_{0}^{\infty}\alpha(\|\mathbf{\phi}(\mathbf{x},t)\|)dt\leq c$, which implies $\lim_{t\rightarrow{+\infty}}\alpha(\|\mathbf{\phi}(\mathbf{x},t)\|)=0$ and $\lim_{t\rightarrow{+\infty}}\mathbf{\phi}(\mathbf{x},t)=\mathbf{0}$. Thus, we have $\mathbf{x}\in S$ and $\Omega_c\subseteq S$. In addition, it is guaranteed that $V(\mathbf{x})\leq\infty$, $\forall \mathbf{x}\in S$. This indicates $\mathbf{x}\in \Omega_{\infty}$ and thus $S\subseteq\Omega_{\infty}$. It follows that $\Omega_c\subseteq S\subseteq\Omega_{\infty}$. Considering that the constant $c$ can be sufficiently large, we have $\lim_{c\rightarrow{+\infty}}\Omega_c=S$.
\end{remark}

\subsection{An upper bound of discretizing error}\label{app:upbound}
An upper bound of the estimation error due to discretizing converse Lyapunov function $V(\mathbf{x})$ is given below.
\begin{lemma}\label{lem:bound}
Let $[\partial \mathbf{f}/\partial \mathbf{x}](\mathbf{0})$ be Hurwitz. There are positive constants $\kappa$, $\eta$, $\lambda$ and $m$ such that the following inequality holds
\begin{equation}\label{bound}
    \left|V(\mathbf{x})-\hat{V}(\mathbf{x})\right|\leq \frac{\kappa n (\Delta t)^3}{12}+\frac{\eta^m \|\mathbf{\phi}(\mathbf{x},t_n)\|^m}{m\lambda}
\end{equation}
\end{lemma}
The proof is given below.
\begin{equation*}
    \begin{split}
        &~~~~V(\mathbf{x})-\hat{V}(\mathbf{x}) \\
        &=V(\mathbf{x})-\int_{0}^{t_n}\alpha(\|\mathbf{\phi}(\mathbf{x},t)\|)dt \\
        & +\int_{0}^{t_n}\alpha(\|\mathbf{\phi}(\mathbf{x},t)\|)dt-\sum_{i=1}^{n}\alpha(\|\mathbf{\phi}(\mathbf{x},t_i)\|)\Delta t,
    \end{split}
\end{equation*}
that
\begin{equation*}
    \begin{split}
        &~~~\left|V(\mathbf{x})-\sum_{i=1}^{n}\alpha(\|\mathbf{\phi}(\mathbf{x},t_i)\|)\Delta t\right| \\
        &\leq\left| V(\mathbf{x})-\int_{0}^{t_n}\alpha(\|\mathbf{\phi}(\mathbf{x},t)\|)dt\right| \\
        & +\left|\int_{0}^{t_n}\alpha(\|\mathbf{\phi}(\mathbf{x},t)\|)dt-\sum_{i=1}^{n}\alpha(\|\mathbf{\phi}(\mathbf{x},t_i)\|)\Delta t\right|\\
        &=\int_{t_n}^{\infty}\alpha(\|\mathbf{\phi}(\mathbf{x},t)\|)dt \\
        & +\left|\int_{0}^{t_n}\alpha(\|\mathbf{\phi}(\mathbf{x},t)\|)dt-\sum_{i=1}^{n}\alpha(\|\mathbf{\phi}(\mathbf{x},t_i)\|)\Delta t\right|\\
    \end{split}
\end{equation*}
By the trapezoidal rule for numerical integration in the interval $[0,t_n]$ \cite{atk08}, the error bound is given by
$$
\left|\int_{0}^{t_n}\alpha(\|\mathbf{\phi}(\mathbf{x},t)\|)dt-\sum_{i=1}^{n}\alpha(\|\mathbf{\phi}(\mathbf{x},t_i)\|)\Delta t\right|\leq \frac{\kappa n (\Delta t)^3}{12}
$$
where $|\partial_{tt}\alpha|\leq\kappa$, $t\in[0,t_n]$. Considering that $[\partial \mathbf{f}/\partial \mathbf{x}](\mathbf{0})$ is Hurwitz, $\mathbf{x=0}$ is an exponentially stable equilibrium point from Corollary $4.3$ in \cite{kha96}. This implies that there are positive constants $\eta$ and $\lambda$, such that $\|\mathbf{\phi}(\mathbf{x},t)\|\leq \eta \|\mathbf{\phi}(\mathbf{x},t_n)\|e^{-\lambda(t-t_n)}$, $\forall t\geq t_n$.
In terms of Condition $3$ of the $\Gamma$ class function in Definition \ref{def}, there is the positive constant $m$ such that
\begin{equation*}
    \begin{split}
        \alpha(\|\mathbf{\phi}(\mathbf{x},t)\|)&\leq\alpha(\eta \|\mathbf{\phi}(\mathbf{x},t_n)\|e^{-\lambda(t-t_n)})\\
        &\leq \eta^m \|\mathbf{\phi}(\mathbf{x},t_n)\|^me^{-m\lambda(t-t_n)}, \quad \forall t\geq t_n
    \end{split}
\end{equation*}
Therefore, we obtain
\begin{equation*}
    \begin{split}
        \int_{t_n}^{\infty}\alpha(\|\mathbf{\phi}(\mathbf{x},t)\|)dt &\leq\int_{t_n}^{\infty}\eta^m \|\mathbf{\phi}(\mathbf{x},t_n)\|^me^{-m\lambda(t-t_n)}dt \\
        &=\eta^m \|\mathbf{\phi}(\mathbf{x},t_n)\|^m\int_{t_n}^{\infty}e^{-m\lambda(t-t_n)}dt \\
        &= \frac{\eta^m\|\mathbf{\phi}(\mathbf{x},t_n)\|^m}{m\lambda}
    \end{split}
\end{equation*}
This completes the proof.

We discuss some notes on the Inequality (\ref{bound}) in the following remark.
\begin{remark}
For the fixed total sampling time $t_n$,  the first term on the right hand side of Inequality (\ref{bound}) converges to $0$ as the sampling time interval $\Delta t$ goes to $0$. Since $\lim_{n\rightarrow{\infty}}\|\mathbf{\phi}(\mathbf{x},t_n)\|=0$, the second term on the right hand side of Inequality (\ref{bound}) converges to $0$ as $n$ goes to the positive infinity.
\end{remark}

\begin{remark}
The first partial derivative of $\alpha(\|\mathbf{\phi}(\mathbf{x},t)\|)$ with respect to time $t$ is given by
\begin{equation*}
    \begin{split}
        \partial_t\alpha &=\frac{\partial\alpha(\|\mathbf{\phi}(\mathbf{x},t)\|)}{\partial t} \\
        &=\alpha'(\|\mathbf{\phi}(\mathbf{x},t)\|)\cdot\frac{\mathbf{\phi}(\mathbf{x},t)^T}{\|\mathbf{\phi}(\mathbf{x},t)\|}\cdot\frac{d\mathbf{\phi}(\mathbf{x},t)}{dt} \\
        &=\alpha'(\|\mathbf{\phi}(\mathbf{x},t)\|)\cdot\frac{\mathbf{\phi}(\mathbf{x},t)^T\mathbf{f}(\mathbf{\phi}(\mathbf{x},t))}{\|\mathbf{\phi}(\mathbf{x},t)\|}
    \end{split}
\end{equation*}
where $\alpha'(z)=d\alpha(z)/dz$. And the second partial derivative of $\alpha(\|\mathbf{\phi}(\mathbf{x},t)\|)$ with respect to time $t$ is given by
\begin{equation}\label{att}
    \begin{split}
        \partial_{tt}\alpha
        &=\frac{\partial^2\alpha(\|\mathbf{\phi}(\mathbf{x},t)\|)}{\partial t^2} \\
        &=\frac{\partial}{\partial t}\left[\alpha'(\|\mathbf{\phi}(\mathbf{x},t)\|)\cdot\frac{\mathbf{\phi}(\mathbf{x},t)^T \mathbf{f}(\mathbf{\phi}(\mathbf{x},t))}{\|\mathbf{\phi}(\mathbf{x},t)\|}\right]\\
        &=\left(\alpha''(\|\mathbf{\phi}(\mathbf{x},t)\|)-1\right)\cdot\frac{[\mathbf{\phi}(\mathbf{x},t)^T \mathbf{f}(\mathbf{\phi}(\mathbf{x},t))]^2}{\|\mathbf{\phi}(\mathbf{x},t)\|^3} \\
        &~+\alpha'(\|\mathbf{\phi}(\mathbf{x},t)\|)\cdot\frac{\|\mathbf{f}(\mathbf{\phi}(\mathbf{x},t))\|^2+\mathbf{\phi}(\mathbf{x},t)^T \partial_t \mathbf{f}}{\|\mathbf{\phi}(\mathbf{x},t)\|},
    \end{split}
\end{equation}
where $\alpha''(z)=d^2\alpha(z)/dz^2$ and
$$
\partial_t \mathbf{f}=\frac{\partial \mathbf{f}(\mathbf{\phi}(\mathbf{x},t))}{\partial t}=\mathbf{f}'(\mathbf{\phi}(\mathbf{x},t)) \mathbf{f}(\mathbf{\phi}(\mathbf{x},t)).
$$
Equation (\ref{att}) enables us to estimate the upper bound of $|\partial_{tt}\alpha|$ in $[0,t_n]$ and obtain the constant $\kappa$.
\end{remark}

\subsection{Information gain}\label{app:gain}

To quantify the reduction in uncertainty on $V(\mathbf{x})$ from observations $\hat{V}(\mathbf{x})$, the information gain is introduced for a Gaussian process as follows
\begin{equation*}
    I(\hat{V}_\mathcal{A};V_\mathcal{A})=\frac{1}{2}\log|I+\sigma^{-2}K_\mathcal{A}|,
\end{equation*}
where $K_\mathcal{A}=[k(\mathbf{x},\mathbf{x}')]_{\mathbf{x},\mathbf{x}'\in \mathcal{A}}$ is the covariance matrix of $V_\mathcal{A}=[V(\mathbf{x})]_{\mathbf{x}\in \mathcal{A}}$ with the sample set $\mathcal{A}$. Let $\gamma_N$ denote the upper bound of $I(\hat{V}_\mathcal{A};V_\mathcal{A})$ for the sample set $\mathcal{A}$ with $|\mathcal{A}|=N$. That is
$$
\gamma_N=\max_{\mathcal{A}\subset X: |\mathcal{A}|=N}I(\hat{V}_\mathcal{A};V_\mathcal{A})
$$
Actually, $\gamma_N$ is related to the choice of kernel functions $k(\mathbf{x},\mathbf{x}')$.
\begin{lemma}\label{gamma}
For $\mathcal{A}\subseteq X\subset R^n$, it holds that
$$
\gamma_N\leq \frac{N}{2\sigma^2}
$$
when $k(\mathbf{x},\mathbf{x}')\leq1$.
\end{lemma}

\begin{proof}
Since $K_\mathcal{A}$ is a positive definite kernel matrix, it follows that
$$
|I+\sigma^{-2}K_\mathcal{A}|=\prod_{i=1}^{N}(1+\sigma^{-2}\lambda_i)
$$
where $\lambda_i$, $i\in I_N=\{1,...,N\}$ are all positive eigenvalues of $K_\mathcal{A}$. Considering that $\log(1+x)\leq x$, $\forall x\geq0$, we have
\begin{equation*}
    \begin{split}
        \log|I+\sigma^{-2}K_\mathcal{A}|&=\log\left(\prod_{i=1}^{N}(1+\sigma^{-2}\lambda_i)\right)\\
        &=\sum_{i=1}^{N}\log(1+\sigma^{-2}\lambda_i)\\
        &\leq \sigma^{-2}\sum_{i=1}^{N}\lambda_i
    \end{split}
\end{equation*}
Because of $\sum_{i=1}^{N}\lambda_i=tr(K_\mathcal{A})$ and $k(\mathbf{x},\mathbf{x}')\leq1$, we obtain
\begin{equation*}
    \sum_{i=1}^{N}\lambda_i=tr(K_\mathcal{A})=\sum_{i=1}^{N}k(\mathbf{x}_i,\mathbf{x}_i)\leq N
\end{equation*}
Therefore, we get
$$
\log|I+\sigma^{-2}K_\mathcal{A}|\leq \sigma^{-2}N, \quad \forall \mathbf{x}\in\mathcal{A}, \quad |\mathcal{A}|=N
$$
and
$$
\gamma_N\leq\frac{N}{2\sigma^2}
$$
This completes the proof.
\end{proof}

\begin{remark}
The upper bound for $\gamma_N$ in Lemma \ref{gamma} is applied to all kernel functions satisfying $k(\mathbf{x},\mathbf{x}')\leq1$. For a specific kernel (e.g., finite dimensional linear kernel, squared exponential kernel and Mat\'ern kernel, etc), the tighter upper bound is available \cite{sri12}.
\end{remark}

\subsection{Computation of the RKHS norm} \label{app:RKHS}
Kernel ridge regression can be formulated as a regularized empirical risk minimization problem over the RKHS $H_k(X)$ as follows
$$
\min_{V\in H_k(X)}\frac{1}{N}\sum_{i=1}^{N}(\hat{V}(\mathbf{x}^{(i)})-V(\mathbf{x}^{(i)}))^2+\theta\|V\|_k^2.
$$
It follows from Theorem $3.4$ in \cite{kan18} that there is a unique solution to the above minimization problem, and the solution is given by $V(\mathbf{x})=\sum_{i=1}^{N}c_i k(\mathbf{x},\mathbf{x}^{(i)})$ with the coefficient vector $\mathbf{c}=(c_1,c_2,...,c_N)^T=(K_N+ N \theta I)^{-1}\hat{V}_N$. Thus, $\|V\|^2_k$ can be computed by
$$
\|V\|^2_k =\langle V,V\rangle_k=\sum_{i=1}^{N}\sum_{j=1}^{N}c_i c_j k(\mathbf{x}^{(i)},\mathbf{x}^{(j)})=\textbf{c}^T K_N \textbf{c}.
$$
By substituting the vector $\textbf{c}=(K_N+ N \theta I)^{-1}\hat{V}_N$ and the positive semidefinite matrix $K_N=P^T\text{diag}(\lambda_i)P$, we have
\begin{equation*}
    \begin{split}
        \|V\|^2_k &=\mathbf{c}^T K_N \mathbf{c} \\
                  &=\hat{V}^T_N (K_N+ N \theta I)^{-1} K_N (K_N+ N \theta I)^{-1}\hat{V}_N \\
                  &=\hat{V}^T_N (K_N+ N \theta I)^{-1} P^T\text{diag}(\lambda_i)P (K_N+ N\theta I)^{-1}\hat{V}_N \\
                  &=\hat{V}^T_NP^T\text{diag}(\lambda_i(\lambda_i+N\theta)^{-2})P\hat{V}_N.
    \end{split}
\end{equation*}

\subsection{Proof of Theorem \ref{theo}}\label{app:theo}
The proof consists of two parts. The first part aims to estimate the ROA for a stable equilibrium point, and the second part centers on the ROA for the case of an existing Lyapunov function $V^\star(\mathbf{x})$.

$(1)$ Since $\mathbf{x=0}$ is an asymptotically stable equilibrium point for the nonlinear system $\mathbf{\dot{x}=f(x)}$, Lemma \ref{lya} guarantees the existence of Lyapunov function $V(\mathbf{x})$. Considering a sequence of sampling points $\mathbf{x}^{(1)}$, $\mathbf{x}^{(2)}$,..., $\mathbf{x}^{(N)}$ in $\mathcal{A}\subset X$ that can generate the stable state trajectory, the values of $V(\mathbf{x}^{(i)})$, $i\in\{1,...,N\}$ can be estimated as $\hat{V}(\mathbf{x}^{(i)})$ according to Equation (\ref{Vest}). By treating $\hat{V}(\mathbf{x}^{(i)})$ as the observations of a GP, the measurement noise $\epsilon$ is uniformly bounded by $\sigma$, which satisfies
$$
 \sigma\leq\frac{\kappa n (\Delta t)^3}{12}+\frac{\eta^m \|\mathbf{\phi}(\mathbf{x},t_n)\|^m}{m\lambda}
$$
from Lemma \ref{lem:bound}. In addition, it follows from Theorem \ref{betaT} that with probability at least $1-\delta$, the inequality
$$
\left|V(\mathbf{x})-\mu_{N-1}(\mathbf{x})\right|\leq \beta^{1/2}_N \sigma_{N-1}(\mathbf{x})
$$
holds, where $\beta_N=2\|V\|_k^2+300\gamma_N\ln^3{(N/\delta)}$ and $\gamma_N\leq\frac{N}{2\sigma^2}$ from Lemma \ref{gamma}.
This allows us to deduce that the inequality
$$
V(\mathbf{x})\leq \mu_{N-1}(\mathbf{x})+\beta^{1/2}_N \sigma_{N-1}(\mathbf{x}), \quad \forall \mathbf{x}\in X
$$
holds with probability at least $1-\delta$.
Then we define a level set for the converse Lyapunov function $V(\mathbf{x})$ as $W_N=\{\mathbf{x}\in X~|~V(\mathbf{x})\leq C_{\max,N}\}$, which includes all the stable state trajectories of sampling points $\mathbf{x}^{(i)}$, $i\in\{1,...,N\}$ and is a compact subset of the region of attraction $S$ (i.e., $W_N\subseteq S$). Thus, if we have $\mu_{N-1}(\mathbf{x})+\beta^{1/2}_N \sigma_{N-1}(\mathbf{x})\leq C_{\max,N}$, the inequality $V(\mathbf{x})\leq C_{\max,N}$ holds with probability at least $1-\delta$, which implies $\text{Prob}(\mathbf{x}\in W_N)\geq 1-\delta$, $\forall \mathbf{x}\in \Omega_{\delta,N}$ with
$$
\Omega_{\delta,N}=\left\{\mathbf{x}\in X ~|~\mu_{N-1}(\mathbf{x})+\beta^{1/2}_N \sigma_{N-1}(\mathbf{x})\leq C_{\max,N}\right\}.
$$
Considering that $W_N\subseteq S$, we obtain $\text{Prob}(\mathbf{x}\in S)\geq 1-\delta$, $\forall \mathbf{x}\in \Omega_{\delta,N}$.

$(2)$ Considering that $V^\star(\mathbf{x})$ is an existing Lyapunov function for the nonlinear system $\mathbf{\dot{x}=f(x)}$ with the certified ROA $\Omega^\star$, it is suggested that the origin is an asymptotically stable equilibrium point. This ensures the existence of a converse Lyapunov function $V(\mathbf{x})$ according to Lemma \ref{lya}. Define $\Delta V(\mathbf{x})=V(\mathbf{x})-V^\star(\mathbf{x})$ as an unknown function for the GP learning. By observing the measurements $\Delta \Hat{V}(\mathbf{x}^{(i)})=\hat{V}(\mathbf{x}^{(i)})-V^\star(\mathbf{x}^{(i)})$ as a GP at the sampling points $\mathbf{x}^{(i)}$, $i\in\{1,...,N\}$ according to GP-ROA based algorithm in Table \ref{tab:gpa}, it follows from Theorem \ref{betaT} that the inequality
$$
\left|\Delta V(\mathbf{x})-\mu_{N-1}(\mathbf{x})\right|\leq \bar{\beta}^{1/2}_N \sigma_{N-1}(\mathbf{x})
$$
holds with probability at least $1-\delta$. By replacing $\Delta V(\mathbf{x})$ with $V(\mathbf{x})-V^\star(\mathbf{x})$, we obtain
$$
\left|V(\mathbf{x})-V^\star(\mathbf{x})-\mu_{N-1}(\mathbf{x})\right|\leq \bar{\beta}^{1/2}_N \sigma_{N-1}(\mathbf{x}).
$$
In light of the proof for a stable equilibrium point, it is concluded that for any $x\in S_{\delta,N}$, the inequality $\text{Prob}(\mathbf{x}\in S)\geq 1-\delta$ holds.
This completes the proof.


\end{document}